\documentclass[conference]{IEEEtran}
\IEEEoverridecommandlockouts
\usepackage{cite}
\usepackage{amsmath,amssymb,amsfonts,amsthm}
\usepackage{mathtools}
\usepackage{algorithmic}
\usepackage[ruled,vlined]{algorithm2e}
\usepackage{graphicx}
\graphicspath{{./Image/}}
\usepackage{textcomp}
\usepackage{xcolor}
\newtheorem{theorem}{Theorem}

\newtheorem{definition}{Definition}
\newtheorem{proposition}{Proposition}
\newtheorem{lemma}[theorem]{Lemma}

\usepackage{booktabs}
\usepackage{siunitx}
\usepackage{soul}

\def\BibTeX{{\rm B\kern-.05em{\sc i\kern-.025em b}\kern-.08em
    T\kern-.1667em\lower.7ex\hbox{E}\kern-.125emX}}
\begin{document}

\title{Fast-Forward Relaying Scheme to Mitigate Jamming Attacks by Full-Duplex Radios}
\author{Vivek Chaudhary and J. Harshan\\
Department of Electrical Engineering,\\
Indian Institute of Technology Delhi, India.
}

\maketitle


\begin{abstract}
In this work, we address reliable communication of low-latency packets in the presence of a full-duplex adversary that is capable of executing a jamming attack while also being able to measure the power levels on various frequency bands. Due to the presence of a strong adversary, first, we point out that traditional frequency-hopping does not help since unused frequency bands may not be available, and moreover, the victim's transition between the frequency bands would be detected by the full-duplex adversary. Identifying these challenges, we propose a new cooperative mitigation strategy, referred to as the Semi-Coherent Fast-Forward Full-Duplex (SC-FFFD) relaying technique, wherein the victim node, upon switching to a new frequency band, seeks the assistance of its incumbent user, which is also a full-duplex radio, to instantaneously forward its messages to the destination using a portion of their powers. Meanwhile, the two nodes cooperatively use their residual powers on the jammed frequency band so as to engage the adversary to continue executing the jamming attack on the same band. Using on-off keying (OOK) and phase-shift-keying (PSK) as the modulation schemes at the victim and the helper node, respectively, we derive upper bounds on the probability of error of jointly decoding the information symbols of the two nodes, and subsequently derive analytical solutions to arrive at the power-splitting factor between the two frequency bands to minimize the error of both the nodes. We also present extensive simulation results for various signal-to-noise-ratio values and PSK constellations to showcase the efficacy of the proposed approach. 
\end{abstract}


\section{Introduction}
\label{sec:introduction}

Wireless applications with low-latency constraints have received traction in the recent past owing to the emergence of vehicular networks, involving autonomous vehicles, Unmanned Aerial Vehicles etc \cite{CCTCY}. While it is imperative to revisit the design of physical-layer algorithms to facilitate low-latency constraints, it is equally important to develop new countermeasures to mitigate Denial-of-Service (DOS) attacks \cite{2evolution_iot} such as jamming, since violation of deadline constraints could lead to catastrophic consequences. Although jamming attack is a well known threat model, and a number of countermeasures have been well studied against it, e.g., Frequency Hopping (FH) \cite{8freq_hopp1}, such traditional mitigation techniques may not be applicable in next-generation networks owing to lack of unused frequency bands due to exponential growth in the number of wireless devices. On the one hand, lack of unused frequency bands certainly poses interesting questions on how to provide ubiquitous and seamless communication of low-latency packets of the victim's node in a frequency band that is already occupied by another node in the network. On the other hand, the very idea of asking the victim node to switch to a new frequency band is questionable especially if the adversary is equipped with sophisticated hardware to execute the jamming attack. For instance, suppose that the adversary, which is equipped with an ideal Full-Duplex (FD) radio, is capable of executing the jamming attack on a frequency band, and is also able to simultaneously measure the power levels on various frequency bands including the one that is jammed. In such a case, the adversary can measure a significant drop in the power levels on the jammed frequency as soon as the victim switches to another frequency band. Therefore, such a reaction may compel the adversary to execute the jamming attack on another frequency band, thereby guaranteeing DOS attack on at least one of the nodes in the network.  

Besides the above observation on the FD adversary, it is clear that due to lack of unused frequency bands the victim node must necessarily share a new frequency band with another node so that the low-latency packets are reliably communicated to the destination within the deadline. Furthermore, the co-existence of the two nodes in the new frequency band must be such that the incumbent user must continue to transmit its information symbols to the destination, and moreover the victim node must also communicate its low-latency packets to reach the destination within the deadline. As a potential solution to achieve the above objective, we propose the use of a FD radio at the incumbent node, which can listen to the messages of the victim node, decode it, and instantaneously forward it to the destination along with its messages. Although one of the challenges of building a FD radio is perfect self-interference-cancellation (SIC), recent technological advancements \cite{13FD2}, \cite{16FD4} have shown promising results towards SIC within desirable limits. Furthermore, apart from the FD features, the prospects of building a fast-forward FD radios have also been explored wherein FD radios can instantaneously process the received symbols and then forward it in the same band.  For instance, in \cite{FFFDR}, the authors were able to achieve near perfect SIC in order of $\mu$s for WiFi signals. Other than the system-related work \cite{FFFDR} on fast-forward relays, several theoretical contributions on fast-forward relays have also been reported in the recent past. For more details, we refer the readers to \cite{HaY} and the references within. 




\subsection{Contributions}

We address a new framework to reliably communicate low-latency messages in the presence of a strong adversarial model wherein the attacker, which is equipped with a FD radio, has the capability to execute a jamming attack on a frequency band while also being able to measure the power levels on a wide range of frequency bands. 
To mitigate the above threat, we present the Semi-Coherent Fast-Forward Full-Duplex (SC-FFFD) relaying technique, wherein the victim node uses $(1- \alpha)$ fraction of its power, for some $0< \alpha < 1$, on a new frequency band to communicate its messages to the destination, while continuing to transmit its residual power on the jammed frequency band. Meanwhile, a full-duplex helper node, which is the incumbent user of the new frequency band, listens to the victim's message, decodes it, and instantaneously forwards it to the destination along with its messages using $\alpha$ fraction of its power. Furthermore, the helper node also pours its residual $1-\alpha$ fraction of power on the jammed frequency band thereby ensuring that the two nodes cooperatively maintain the same power levels on both the frequency bands. With such a strategy, the helper node assists the victim's message to reach the destination without violating the latency constraints.
Using On-Off Keying (OOK) and Phase-Shift-Keying (PSK) as the modulation schemes at the victim and the helper node, respectively, we present a thorough analysis on the error performance of the SC-FFFD technique when the destination employs a joint decoder on the new frequency band. We derive upper bounds on the average probability of error of the joint decoder at high signal-to-noise-ratio (SNR) values, and subsequently identify dominant error terms as a function of $\alpha$, henceforth referred to as the power-splitting factor. Finally, we prove non-trivial relations between the dominant terms to determine an appropriate value of $\alpha$ that minimizes the average probability of error of the joint decoder. Through extensive simulations, we show that the average probability of error of the SC-FFFD technique decreases with increasing SNR, which in turn implies that the victim node can reliably communicate its messages to the destination.

Although \cite{20DF, FFFDR, improved_sec, Cog_FH_FD} have studied jamming aspects with relaying techniques and FD radios, they have not addressed the idea of fast-forward relaying to engage an FD jammer on one frequency band. Among these prior works, the contributions of \cite{Cog_FH_FD} is closest to our work. However, unlike our work, \cite{Cog_FH_FD} assumes sufficient number of unused frequency bands to execute FH as countermeasure, and moreover, their mitigation technique does not engage the jammer on one frequency band. Throughout this paper, we refer to the victim, the helper, the attacker and the destination as Alice, Charlie, Dave and Bob, respectively.


\section{System Model and Problem Statement}
\label{sec:system_model}

Consider a network model, as shown in Fig. \ref{FFFD_SM}, consisting of two nodes, namely Alice and Charlie, that communicate with a base station, namely Bob, using orthogonal frequencies, represented by $f_{AB}$ and $f_{CB}$, respectively. We assume that unused frequency bands are not available as the network is operating at capacity with maximum number of users. The network requirements of Alice and Charlie are heterogeneous in the sense that Alice is interested in communicating low-rate messages that have low-latency constraints, whereas Charlie is interested in communicating high-rate messages that need not satisfy any low-latency constraints. The network also includes an active adversary, namely, Dave, that injects high-powered noise signals on $f_{AB}$ to execute a DOS attack on the low-latency messages of Alice. A key feature of the attack model is that Dave is equipped with a FD radio with perfect SIC capability such that it can scan a wide range of spectrum to measure the average power levels including $f_{AB}$ and $f_{CB}$. With the existence of such a vigilant jammer, Alice must somehow mitigate this jamming attack so as to continue transmitting her low-latency messages to Bob. Although a straightforward mitigation strategy for Alice is to hop to another frequency band, such a strategy would assist Dave to identify a significant drop in the power levels on $f_{AB}$. This \emph{frequency hole} on $f_{AB}$ would further compel Dave to attack on one of the remaining frequency bands resulting in degradation of error performance of at least one of the nodes in the network. Therefore, while it is necessary for Alice to hop to another frequency band, Dave must neither observe a dip in the power levels on $f_{AB}$, nor observe a surge in the power levels of another frequency band. Furthermore, Alice must not communicate any pilots on the new frequency band since the communication-overhead in pilot transmission does not help the low-latency constraints on the packets.

In the next section, we present a new cooperative relaying strategy wherein Alice seeks the help of Charlie (in the vicinity) to communicate her low-latency messages to Bob. 

\section{Semi-Coherent Fast-Forward Full-Duplex Relaying Technique}
\label{sec:FF}

\begin{figure}[t]
\centering
\includegraphics[scale = 0.22]{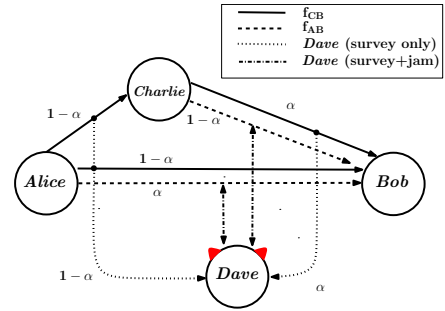}
\vspace{-0.2cm}
\caption{\label{FFFD_SM}Depiction of the system model comprising Alice (the victim node), Charlie (the helper node), Bob (the destination), and Dave (the jammer with FD radio). Upon experiencing a jamming attack on $f_{AB}$, Alice seeks the assistance of Charlie, which has a FD radio, to instantaneously forward her messages to Bob on frequency band $f_{CB}$.}
\end{figure}

As a countermeasure to mitigate the jamming attack, Bob directs Alice to switch to the frequency band $f_{CB}$, which is already used by Charlie. Furthermore, Bob guides Charlie to continue operating on $f_{CB}$, while also requesting him to relay Alice's information symbols in the FD mode. Since Charlie is capable of fast-forward relaying, he listens to the transmission of Alice on $f_{CB}$, decodes her information symbol, and then instantaneously forwards the decoded symbol to Bob by appropriately embedding it with its message using physical-layer techniques. As a consequence, Bob witnesses a multiple access channel on the frequency band $f_{CB}$ by receiving a linear combination of symbols from both Alice and Charlie. While this idea of fast-forward relaying serves Alice's messages to reach Bob with no additional delay, it is important to note that Dave will now observe zero transmission power in the frequency band $f_{AB}$. To circumvent this problem, we propose a power-splitting strategy between the two nodes, wherein Alice and Charlie employ $1- \alpha$ and $\alpha$ fractions of their power on the frequency band $f_{CB}$, respectively, for some $\alpha \in (0, 1)$. Meanwhile, they use their residual powers of $\alpha$ and $1-\alpha$ fractions on $f_{AB}$. As a result of this strategy, the attacker Dave neither observes a surge in the power level on $f_{CB}$ nor a dip in the power level on $f_{AB}$, thus deceiving Dave to believe that Alice has continued to transmit on $f_{AB}$. Henceforth, throughout this work, we refer to this strategy as the Fast-Forward Full-Duplex (FFFD) relaying scheme.

Under the framework of FFFD relaying scheme, we are interested in a semi-coherent (SC) modulation scheme, wherein Alice employs a non-coherent modulation technique, e.g. OOK, and Charlie employs a conventional coherent modulation technique, e.g., PSK, QAM. We highlight that the use of non-coherent modulation technique at Alice is to reduce the communication-overhead of transmitting pilot symbols upon switching to the frequency band $f_{CB}$, thereby facilitating the transmission of low-latency messages to Bob. In the next section, we present a detailed explanation on the signaling scheme of the SC-FFFD protocol. 
 
\subsection{Signal Model of SC-FFFD Relaying Protocol} 
 
In the proposed SC-FFFD relaying protocol, Alice employs OOK, denoted by the constellation $\mathcal{S}_{A} = \{0, 1\}$, whereas Charlie uses the traditional $M$-PSK constellation, denoted by the constellation $\mathcal{S}_{C} = \{e^{\frac{\iota 2\pi (j + 0.5)}{M}} ~|~ j = 0, 1, \ldots, M-1\}$, where $\iota = \sqrt{-1}$, and $M = 2^{m}$, for some positive integer $m$. As highlighted earlier, Charlie uses a full-duplex radio with perfect SIC capability. Upon transmission of information symbol from Alice, Charlie decodes Alice's symbols, and depending on the decoded bit, he instantaneously  transmits a modified version of the $M$-PSK symbol so that Bob can jointly decode the information symbols of both Alice and Charlie. In particular, if $x \in \mathcal{S}_{A}$ is transmitted from Alice on $f_{CB}$, Charlie receives
\begin{equation*}
r_{C} = \sqrt{1-\alpha} h_{AC}x + n_{C},
\end{equation*}
where $1 - \alpha$ is the associated power when transmitting symbol $1$, the complex number $h_{AC}\sim {\cal CN}(0,\sigma_{AC}^2)$ is the baseband channel between Alice and Charlie, and $n_C \sim {\cal CN}(0,N_o)$ is the additive white Gaussian noise (AWGN) at Charlie. Due to proximity between Alice and Charlie, we assume $\sigma_{AC}^2 \geq 1$. Owing to no knowledge of the instantaneous channel realization $h_{AC}$, Charlie performs non-coherent energy detection to obtain an estimate of $x$, denoted by $\hat{x}_{C}$. Furthermore, in order to transmit its own information symbol $y \in  \mathcal{S}_{C}$, Charlie transmits 
\begin{equation*}
\left\{ \begin{array}{llllll}
y, & \mbox{ if } \hat{x}_{C} = 0;\\
\sqrt{\alpha} e^{\frac{\iota \pi}{M}} y, & \mbox{ if } \hat{x}_{C} = 1.\\
\end{array}
\right.
\end{equation*}
With that, Charlie transmits a symbol from either $\mathcal{S}_{C}$ or $\sqrt{\alpha}e^{\frac{\iota \pi}{M}}\mathcal{S}_{C}$ at a given round of transmission. As a result of this instantaneous processing at Charlie, the baseband symbol received at Bob is of the form

\begin{small}
\begin{equation}
\label{eq:rx_symbool_bob}
r_{B} = \left\{ \begin{array}{llllllllll}
h_{CB}\ y + n_{B}, & \mbox{ if } x = 0, \hat{x}_{C} = 0;\\
h_{CB}\ \sqrt{\alpha} e^{\frac{\iota \pi}{M}} y + n_{B}, & \mbox{ if } x = 0, \hat{x}_{C} = 1;\\
h_{CB}\ y + h_{AB}\sqrt{1-\alpha} + n_{B}, & \mbox{ if } x = 1, \hat{x}_{C} = 0;\\
h_{CB}\ \sqrt{\alpha} e^{\frac{\iota \pi}{M}} y + h_{AB}\sqrt{1-\alpha} + n_{B}, & \mbox{ if } x = 1, \hat{x}_{C} = 1;
\end{array}
\right.
\end{equation}
\end{small}
\indent where $h_{CB} \sim {\cal CN}(0, 1)$ is the baseband channel between Charlie and Bob, $h_{AB} \sim {\cal CN}(0, 1)$ is the baseband channel from Alice to Bob, and $n_{B} \sim \mathcal{CN}(0, N_{o})$ is the AWGN at Bob. Since Charlie communicates with Bob using coherent signaling method, we assume that Bob has perfect knowledge of the channel realization $h_{CB}$. Since Alice has shifted her frequency to $f_{CB}$ as a reaction against jamming, and no pilots are communicated on $f_{CB}$, we assume that Bob has no knowledge $h_{AB}$. We assume that all the channel realizations and additive noise components are statistically independent. Henceforth, throughout this paper, we denote $\mbox{SNR} = \frac{1}{N_{o}}$. 

With the signal model in \eqref{eq:rx_symbool_bob}, Bob needs to decode the information symbols of both Alice and Charlie. To assist joint detection of information symbols of both nodes, Alice's information symbol can be recovered by observing two metrics: (i) whether $r_{B}$ is closer to a point in $\sqrt{\alpha}e^{\frac{\iota \pi}{M}}\mathcal{S}_{C}$ instead of a point in $\mathcal{S}_{C}$, and (ii) whether the energy contributed by the effective additive noise is $1-\alpha + N_{o}$ instead of $N_{o}$. Overall, this framework of joint decoding of the information symbols of Alice and Charlie corresponds to applying a combination of coherent and non-coherent decoding mechanism on an equivalent multiple access channel model induced by the SC-FFFD relaying protocol. 

\subsection{Observation on Power Measurements at Dave}

In the FFFD relaying protocol, both Alice and Charlie communicate simultaneously on the frequency band $f_{CB}$. As a result, with the assumption that the decoding error introduced at Charlie is negligible, the average power measured on $f_{CB}$ is unity irrespective of whether Alice transmits symbol 1 or symbol 0. Meanwhile, whenever Alice transmits symbol $1$, upon correctly decoding it at Charlie, in the FFFD protocol, Charlie transmits its residual power $(1-\alpha)$ on $f_{AB}$. Concurrently, Alice also transmits its residual power $\alpha$ on $f_{AB}$, and as a result, the total average power observed on $f_{AB}$ continues to be one and zero when symbol 1 and symbol 0 is transmitted by Alice, respectively. This implies that upon power measurements at Dave, the power levels measured on $f_{CB}$ continues to be unity, which is same as the power measured before Charlie helped Alice. Similarly, the power levels measured on $f_{AB}$ continues to be that of OOK, which is same as the power measured on $f_{AB}$ before Charlie helped Alice. 

\begin{figure*}
\begin{small}
\begin{eqnarray}
\frac{dP_{11}}{d\alpha}  = \left(\frac{N_{C0}}{N_{C0} + \sigma_{AC}^2(1-\alpha)}\right)^{\frac{N_{C0} + \sigma_{AC}^2(1-\alpha)}{\sigma_{AC}^2(1-\alpha)}} \times 
\frac{\sigma_{AC}^2(1-\alpha)\left(ln\left[\frac{N_{C0}}{N_{C0}+\sigma_{AC}^2(1-\alpha)}\right] + 1\right) + N_{C0} ln\left[\frac{N_{C0}}{N_{C0}+\sigma_{AC}^2(1-\alpha)}\right]}{\sigma_{AC}^2(1-\alpha)^2}\label{eq_diff_P11}.
\end{eqnarray}
\end{small}
\hrule
\end{figure*}



\subsection{Error performance at Charlie.}

With non-coherent energy detection, Charlie makes a decision using the likelihood ratio $f(r_{C}|x=0) \underset{1}{\overset{0}{\gtrless}} f(r_{C}|x=1),$ where $f(r_{C}|x=i)$ is the probability density function of $r_{C}$ conditioned on $x = i$, for $i \in \mathcal{S}_{A}$. With that, the probability of decoding symbol $0$ as symbol $1$ at Charlie is given by
\begin{equation*}
P_{01} = \mbox{Pr}\left\{\left.\frac{\frac{1}{\pi N_{C0}} e^{-\frac{|r_C|^2}{N_{C0}}}}{\frac{1}{\pi N_{C1}} e^{-\frac{|r_C|^2}{N_{C1}}}}<1\right\vert r_C = n_{C}\right\} = e^{-\frac{\beta}{N_{C0}}},
\end{equation*}
where $N_{C0} = N_{o}$, $N_{C1} = \sigma^{2}_{AC}(1-\alpha) + N_{o}$, $\beta = \frac{N_{C0} N_{C1}}{N_{C0} - N_{C1}} ln\left(\frac{N_{C0}}{N_{C1}}\right)$. Thus, the probability of correct decoding of symbol $0$ is given by $P_{00} = 1-P_{01}= 1-e^{-\frac{\beta}{N_{C0}}}$. Similarly, the probability of decoding symbol $1$ as symbol $0$ can be computed as
\begin{eqnarray*}
P_{10} & = & \mbox{Pr}\left\{\left.\frac{\frac{1}{\pi N_{C1}} e^{-\frac{|r_C|^2}{N_{C1}}}}{\frac{1}{\pi N_{C0}} e^{-\frac{|r_C|^2}{N_{C0}}}}<1\right\vert r_C =h_{AC}\sqrt{1-\alpha} + n_C\right\}\label{FD_P10},\\
& = & 1-e^{-\frac{\beta}{N_{C1}}}.
\end{eqnarray*}
Thus, the probability of correct decoding of symbol $1$ is given by $P_{11} = 1-P_{10}=e^{-\frac{\beta}{N_{C1}}}$. Note that both $P_{01}$ and $P_{10}$ are functions of $\alpha$. In the following lemmas, we present some insights on $P_{01}$ and $P_{10}$.

\begin{lemma}\label{L_p11}
The term $P_{11}$ decreases as $\alpha$ increases in the interval $(0, 1)$.
\end{lemma}
\begin{proof}
The term $P_{11}$ can be rewritten as
\begin{equation}
\label{p11_sim}
P_{11} = \left(\frac{N_{C0}}{N_{C0} + \sigma_{AC}^2(1 -\alpha)}\right)^{\frac{N_{C0}}{\sigma_{AC}^2(1-\alpha)}}.
\end{equation}
Differentiating $P_{11}$ w.r.t. $\alpha$, we get \eqref{eq_diff_P11}. If we closely observe the numerator of the second term, we have $N_{C0} ln\left[\frac{N_{C0}}{N_{C0}+\sigma_{AC}^2(1-\alpha)}\right]<0$, and for $\sigma_{AC}^2(1-\alpha)\left(ln\left[\frac{N_{C0}}{N_{C0}+\sigma_{AC}^2(1-\alpha)}\right] + 1\right)<0$, this implies $\alpha \leq \alpha_{1} \triangleq 1-1.71\frac{N_{C0}}{\sigma_{AC}^2}$. Let $\alpha_1^+ = 1-\omega\frac{N_{C0}}{\sigma_{AC}^2}$ such that $0 < \omega< 1.71$. Substituting $\alpha_1^+$ in \eqref{p11_sim} yields $\frac{1}{(1+\omega)^{\frac{1}{\omega}}}$, which is a decreasing function when $0 < \omega <1.71$. Therefore, $P_{11}$ decreases with $\alpha$ in the interval $(0, 1)$.

\end{proof}

\begin{lemma}\label{lemma_p11_p00_domination}
The expressions for $P_{01}$ and $P_{10}$ are such that $P_{11} > P_{10}$ for $\alpha \in (0, 1 - \frac{N_{C0}}{\sigma^{2}_{AC}})$, and $P_{00} > P_{01}$ for $\alpha \in (0, 1)$.
\end{lemma}
\begin{proof}
Since $P_{10}$ is an increasing function of $\alpha$, we are interested in computing the value of $\alpha$ for which $P_{10} = 0.5$. The expression for $P_{10}$ can be rewritten as $1 - \left(\frac{N_{C0}}{N_{C1}}\right)^{\frac{N_{C0}}{N_{C1}-N_{C0}}}$. This implies that $\left(\frac{N_{C0}}{N_{C0} + \sigma^{2}_{AC}(1 - \alpha)}\right)^{\frac{N_{C0}}{\sigma^{2}_{AC}(1-\alpha)}} = 0.5$ only when $\sigma^{2}_{AC}(1 - \alpha) = N_{C0}$. Therefore, until $\alpha < 1 - \frac{N_{C0}}{\sigma^{2}_{AC}}$, the term $P_{11}$ dominates $P_{10}$. Note that at high SNR, i.e., when $N_{C0} << 1$, this implies that $P_{11}$ dominates $P_{10}$ in the interval $\alpha \in (0, 1 - \epsilon)$, where $\epsilon$ is a negligible number. For the second result of this lemma,  the expression for $P_{01}$ can be rewritten as 
\begin{equation*}
\left(\frac{N_{C0}}{N_{C0} + \sigma^{2}_{AC}(1-\alpha)}\right)^{\frac{N_{C0} + \sigma^{2}_{AC}(1-\alpha)}{\sigma^{2}_{AC}(1-\alpha)}}.
\end{equation*}
Denoting $\sigma^{2}_{AC} (1- \alpha) = \delta N_{C0}$, for $\delta > 0$, the above expression can be written as 
\begin{equation*}
\left(\frac{1}{1+ \delta}\right)^{\frac{1 + \delta}{\delta}} = \left(\frac{1}{1+ \delta}\right)\frac{1}{(1 + \delta)^{\frac{1}{\delta}}} < \left(\frac{1}{1+ \delta}\right)\frac{1}{2},
\end{equation*}
where the inequality follows because of the lower bound $(1 + \delta)^{\frac{1}{\delta}} > 2$ when $\delta > 0$. This implies that $P_{01}$ never hits $0.5$, and therefore, $P_{00}$ dominates $P_{01}$ in the interval $\alpha \in (0, 1)$. 
\end{proof}

\begin{lemma}\label{P01_P10}
For any $0\leq\alpha\leq1$, we have $P_{01} < P_{10}$.
\end{lemma}
\begin{proof}
The ratio $\frac{P_{10}}{P_{01}}$ is given by
\begin{eqnarray*}
\frac{P_{10}}{P_{01}} & = & \frac{1-e^{-\frac{\beta}{N_{C1}}}}{e^{-\frac{\beta}{N_{C0}}}}=\left(\frac{N_{C1}}{N_{C0}}\right)^{\frac{N_{C1}}{N_{C1}-N_{C0}}}-\left(\frac{N_{C1}}{N_{C0}}\right),\\
& = & \frac{N_{C1}}{N_{C0}}\left[\left(\frac{N_{C1}}{N_{C0}}\right)^{\frac{N_{C0}}{N_{C1}-N_{C0}}}-1\right],\\
& = & (1 + \delta)\left((1 + \delta)^{\frac{1}{\delta}} - 1\right),
\end{eqnarray*}
where the last equality is written by substituting $\sigma^{2}_{AC}(1-\alpha) = \delta N_{C0}$, for $\delta > 0$. Since $N_{C1}\geq N_{C0}$, we have $(1 + \delta)^{\frac{1}{\delta}} > 2$, and therefore, we conclude that $P_{01} <  P_{10}$.
\end{proof}

Having understood the behavior of $P_{11}$ and $P_{00}$ as a function of $\alpha$, we proceed to analyze the error performance of jointly decoding the information symbols of Alice and Charlie at Bob.

\subsection{Error performance at Bob}

Based on the signal model in \eqref{eq:rx_symbool_bob}, it is clear that Bob has to make use of a combination of coherent and non-coherent detection method to jointly decode the information symbols of Alice and Charlie. It is worthwhile to note that when Alice transmits symbol $1$, the symbol received at Bob has higher noise variance as compared to when symbol $0$ was transmitted. Therefore, we represent the effective noise variance at Bob as $N_{B0} = N_o$ when Alice transmits symbol $0$, and also $N_{B1} = N_o + (1-\alpha)$ when Alice transmits symbol $1$. To arrive at $N_{B1}$, we have used the fact that $h_{AB} \sim\mathcal{CN}(0, 1)$. From first principles, a Maximum A Posteriori (MAP) detector to jointly decode $i \in \{0, 1\}$ for OOK, and $j \in \{0, 1, \ldots, M-1\}$ for the PSK symbol $e^{\frac{\iota 2\pi (j + 0.5)}{M}}$, for the above signal model is given by
\begin{equation}
\label{eq:joint_MAP_Bob}
\hat{i},\hat{j} = \arg \max_{i, j} g\left(r_B|x = i, y = e^{\frac{\iota 2 \pi (j + 0.5)}{M}}, h_{CB} \right),
\end{equation}
where $g\left(r_B|x = i, y = e^{\frac{\iota 2 \pi (j + 0.5)}{M}}, h_{CB} \right)$ is the probability density function of $r_{B}$ subject to a given realizations of $i$ and $j$, and also the realization of the channel $h_{CB}$. Since Charlie may also add error events when decoding $x$, the conditional density function of $r_{B}$ is a Gaussian mixture weighed by the probabilities of decoding error at Charlie. In particular, we have $g\left(r_B|x = 0, y = e^{\frac{\iota 2 \pi (j + 0.5)}{M}}, h_{CB} \right)$ given by
\begin{eqnarray*}
P_{00} f_{0}(r_{B}| t = e^{\frac{\iota 2 \pi (j + 0.5)}{M}}, h_{CB}) + \\
\label{eq:gm_1}
~P_{01}f_{0}(r_{B}| t = \sqrt{\alpha} e^{\frac{\iota \pi }{M}} e^{\frac{\iota 2 \pi (j + 0.5)}{M}}, h_{CB})
\end{eqnarray*}
such that $f_{0}(r_{B}|t, h_{CB}) = \frac{1}{\pi N_{B0}}e^{-\frac{|r_{B} - h_{CB}t|^2}{N_{B0}}}$. Similarly, we have $g\left(r_B|x = 1, y = e^{\frac{\iota 2 \pi (j + 0.5)}{M}}, h_{CB} \right)$ given by
\begin{eqnarray*}
\label{eq:gm_2}
P_{10} f_{1}(r_{B}| t = e^{\frac{\iota 2 \pi (j + 0.5)}{M}}, h_{CB}) + \nonumber\\
~P_{11}f_{1}(r_{B}| t = \sqrt{\alpha} e^{\frac{\iota \pi }{M}} e^{\frac{\iota 2 \pi (j + 0.5)}{M}}, h_{CB}),
\end{eqnarray*}
such that $f_{1}(r_{B}|t, h_{CB}) = \frac{1}{\pi N_{B1}}e^{-\frac{|r_B-h_{CB}t|^2}{N_{B1}}}$. 

It is straightforward to note that the average probability of error of the joint MAP decoder in \eqref{eq:joint_MAP_Bob} is a function of $\alpha$ since the intra-distance properties of the constellation observed by Bob varies with $\alpha$. Therefore, an important task is to compute $\alpha \in (0, 1)$ that minimizes this average probability of error. However, we notice that evaluating the average probability of error of the MAP decoder is non-trivial mainly due to the intricacies involved in handling Gaussian mixtures. Towards obtaining a near-optimal solution, we present an approximation on the MAP decoder, and subsequently compute the value of $\alpha$ that minimizes the probability of error of the sub-optimal decoder 

\section{Fast Forward Full Duplex Dominant Decoder}
\label{sec:subopt_decoder}

When handling Gaussian mixture as \emph{a priori} probability density function in the MAP detector of \eqref{eq:joint_MAP_Bob}, it is well known that the decoding metric \cite[Section II.A]{RaM}
\begin{equation}
\label{eq:joint_MAP_Bob_max}
\hat{i},\hat{j} = \arg \max_{i, j} g_{m}\left(r_B|x = i, y = e^{\frac{\iota 2 \pi (j+ 0.5)}{M}}, h_{CB} \right),
\end{equation}
provides near-optimal error performance where $g_{m}\left(r_B|x = i, y = e^{\frac{\iota 2 \pi (j+ 0.5)}{M}}, h_{CB} \right)$ is given in \eqref{eq:gm_1_max} and \eqref{eq:gm_2_max}, for $i = 0$ and $i = 1$, respectively. 
\begin{figure*}
\begin{eqnarray}
g_{m}\left(r_B|x = 0, y = e^{\frac{\iota 2 \pi (j+ 0.5)}{M}}, h_{CB} \right) = \max \left(P_{00} f_{0}(r_{B}| t = e^{\frac{\iota 2 \pi (j+ 0.5)}{M}}, h_{CB}), P_{01}f_{0}(r_{B}| t = \sqrt{\alpha} e^{\frac{\iota \pi }{M}} e^{\frac{\iota 2 \pi (j+ 0.5)}{M}}, h_{CB})\right)
\label{eq:gm_1_max}
\end{eqnarray}
\begin{eqnarray}
\label{eq:gm_2_max}
g_{m}\left(r_B|x = 1, y = e^{\frac{\iota 2 \pi (j+ 0.5)}{M}}, h_{CB} \right) = \max \left(P_{10} f_{1}(r_{B}| t = e^{\frac{\iota 2 \pi (j+ 0.5)}{M}}, h_{CB}), P_{11}f_{1}(r_{B}| t = \sqrt{\alpha} e^{\frac{\iota \pi }{M}} e^{\frac{\iota 2 \pi (j+ 0.5)}{M}}, h_{CB})\right)
\end{eqnarray}
\hrule
\end{figure*}
Furthermore, in \eqref{eq:gm_1_max} and \eqref{eq:gm_2_max}, we note that the terms $P_{11}$ and $P_{00}$ respectively dominate $P_{10}$ and $P_{01}$ due to the results in Lemma \ref{lemma_p11_p00_domination}. As a result, we present a sub-optimal decoder in the following definition.

\begin{definition}
Using the results in Lemma \ref{lemma_p11_p00_domination}, the decoding metric in \eqref{eq:joint_MAP_Bob_max} can be further reduced by dropping the terms with $P_{01}$ and $P_{10}$ as
\begin{equation}
\label{eq:joint_MAP_Bob_dominant}
\hat{i},\hat{j} = \arg \max_{i, j} g_{a}\left(r_B|x = i, y = e^{\frac{\iota 2 \pi (j+ 0.5)}{M}}, h_{CB} \right),
\end{equation}
where 
\begin{eqnarray*}
\label{eq_pdf1}
g_{a}\left(r_B|x = 0, y, h_{CB} \right) & = & \frac{P_{00}}{\pi N_{B0}}e^{-\frac{|r_{B} - h_{CB}\ y|^2}{N_{B0}}},\\
\label{eq_pdf2}
g_{a}\left(r_B|x = 1, y, h_{CB} \right) & = & \frac{P_{11}}{\pi N_{B1}}e^{-\frac{|r_{B} - h_{CB}\ \sqrt{\alpha}e^{\frac{\pi}{M}}y|^2}{N_{B1}}}.
\end{eqnarray*}
\end{definition}
Henceforth, throughput the paper, we refer to the decoder in \eqref{eq:joint_MAP_Bob_dominant} as the Fast-Forward Full-Duplex Joint Dominant (FFFD-JD) decoder. Based on the decoding metric in \eqref{eq:joint_MAP_Bob_dominant}, Bob uses $r_{B}$ to decode to a point in the constellation $\mathcal{S}_{C} \cup \sqrt{\alpha} e^{\frac{\iota \pi}{M}}\mathcal{S}_{C} \subset \mathbb{C}$, wherein the Gaussian distribution centered around the points in $\mathcal{S}_{C}$ has variance $N_{B0}$, whereas the Gaussian distribution centered around the points in $\sqrt{\alpha} e^{\frac{\iota \pi}{M}}\mathcal{S}_{C}$ has variance $N_{B1}$. For instance, an example for the constellation $\mathcal{S}_{C} \cup \sqrt{\alpha} e^{\frac{\iota \pi}{M}}\mathcal{S}_{C}$ with $M = 4$ is shown in Fig. \ref{FFFD_CN_1}, where the set of circles denote $\mathcal{S}_{C}$ and the set of diamonds denote $\sqrt{\alpha} e^{\frac{\iota \pi}{M}}\mathcal{S}_{C}$. 

In the next section, we compute upper bounds on the probability of error of jointly decoding the symbols of OOK and PSK constellation using the FFFD JD decoder. Subsequently, we use the upper bound to recover an appropriate value of $\alpha \in (0, 1)$ that minimizes the average probability of error of the FFFD JD decoder. 

\subsection{Error Performance of FFFD Joint Dominant Decoder}

With FFFD JD decoder as given in \eqref{eq:joint_MAP_Bob_dominant}, a pair $(i, j) \in \{0, 1\} \times \{0, 1, \ldots, M-1\}$ can be incorrectly decoded as $(\bar{i}, \bar{j})$ such that $(\bar{i}, \bar{j}) \neq (i, j)$ if 
\begin{equation*}
\Delta_{(i, j) -> (\bar{i}, \bar{j})} \triangleq \frac{g_{a}\left(r_B|x = \bar{i}, y = e^{\frac{\iota 2 \pi (\bar{j} + 0.5)}{M}}, h_{CB} \right)}{g_{a}\left(r_B|x = i, y = e^{\frac{\iota 2 \pi (j+ 0.5)}{M}}, h_{CB} \right)} \geq 1,
\end{equation*}
where $\Delta_{(i, j) -> (\bar{i}, \bar{j})}$ is the error event. Furthermore, if Alice and Charlie have chosen the pair $(i = 0, j)$, the probability that Bob incorrectly decodes to another pair $(\bar{i}, \bar{j})$, denoted by $\mbox{Pr}\left((0, j) \rightarrow (\bar{i}, \bar{j}) \right)$, is given by
\begin{eqnarray}
\label{eq:pair_wise_1}
\mbox{Pr}\left((0, j) \rightarrow (\bar{i}, \bar{j}) \right) = P_{00} \mbox{Pr}\left(\Delta_{(0, j) \rightarrow (\bar{i}, \bar{j})} \geq 1 \right\vert r_{B} = r_{00})\nonumber \\
+ ~P_{01} \mbox{Pr}\left(\Delta_{(0, j) \rightarrow (\bar{i}, \bar{j})} \geq 1 \right\vert r_{B} = r_{01}),
\end{eqnarray}
where $r_{00} = h_{CB}e^{\frac{\iota 2 \pi (j+ 0.5)}{M}} + n_{B}$ and $r_{01} = h_{CB}\sqrt{\alpha} e^{\frac{\iota \pi}{M}} e^{\frac{\iota 2 \pi (j+ 0.5)}{M}} + n_{B}$. Similarly, if Alice and Charlie have chosen the pair $(i = 1, j)$, the probability that Bob incorrectly decodes to another pair $(\bar{i}, \bar{j})$, denoted by $\mbox{Pr}\left((1, j) \rightarrow (\bar{i}, \bar{j}) \right)$, is given by
\begin{eqnarray}
\label{eq:pair_wise_2}
\mbox{Pr}\left((1, j) \rightarrow (\bar{i}, \bar{j}) \right) = P_{11} \mbox{Pr}\left(\Delta_{(1, j) \rightarrow (\bar{i}, \bar{j})} \geq 1 \right\vert r_{B} = r_{11})\nonumber \\
+ ~P_{10} \mbox{Pr}\left(\Delta_{(1, j) \rightarrow (\bar{i}, \bar{j})} \geq 1 \right\vert r_{B} = r_{10}),
\end{eqnarray}
where $r_{11} = h_{CB}\sqrt{\alpha} e^{\frac{\iota \pi}{M}} e^{\frac{\iota 2 \pi (j+ 0.5)}{M}} + h_{AB}\sqrt{1 -\alpha} + n_{B}$ and $r_{10} = h_{CB}e^{\frac{\iota 2 \pi (j+ 0.5)}{M}} + h_{AB}\sqrt{1 -\alpha} + n_{B}$. To compute $\mbox{Pr}\left((i, j) \rightarrow (\bar{i}, \bar{j}) \right)$, we have considered the error events when decoding Alice's symbols at Charlie. Overall, for a given $h_{CB}$, the probability of error of the decoder in \eqref{eq:joint_MAP_Bob_dominant} is given by
\begin{equation}
\label{eq:dominant_error_with_hCB}
\mbox{Pr}(error|h_{CB}) = \frac{1}{2M} \sum_{(i, j)} \mbox{Pr}((\hat{i}, \hat{j}) \neq (i, j) |(i, j)),
\end{equation}
where $\mbox{Pr}((\hat{i}, \hat{j}) \neq (i, j) |(i, j))$ is the probability that Bob decodes to a pair other than $(i, j)$, when $(i, j)$ is chosen by Alice and Charlie. Furthermore, using union bound, we have
\begin{equation}
\label{eq:union_bound_pairwise}
\mbox{Pr}((\hat{i}, \hat{j}) \neq (i, j) |(i, j)) \leq \sum_{(\bar{i}, \bar{j}) \neq (i, j)} \mbox{Pr}\left((i, j) \rightarrow (\bar{i}, \bar{j}) \right),
\end{equation}
where $\mbox{Pr}\left((0, j) \rightarrow (\bar{i}, \bar{j}) \right)$ and $\mbox{Pr}\left((1, j) \rightarrow (\bar{i}, \bar{j}) \right)$ are given in \eqref{eq:pair_wise_1} and \eqref{eq:pair_wise_2}, respectively.

\begin{figure}[t]
\centering
\includegraphics[scale=0.5]{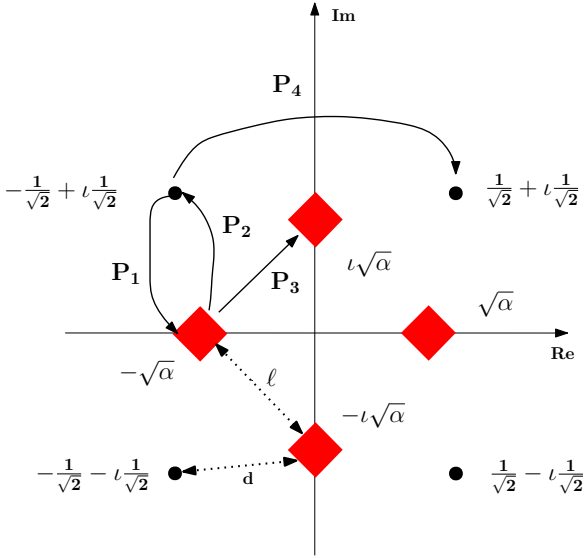}
\vspace{-0.2cm}
\caption{An example for the two-dimensional constellation observed by Bob as a result of the SC-FFFD relaying protocol. With $M = 4$, the set of black circles represent the constellation $S_{C}$ used by Charlie upon decoding symbol $0$ from Alice, whereas the set of red diamonds represent the scaled and rotated version of the constellation $\mathcal{S}_{C}$ used by Charlie upon decoding symbol $1$ from Alice.}
\label{FFFD_CN_1}
\end{figure}

In the following theorem, we present high SNR approximations on $\mbox{Pr}(error|h_{CB})$ given in \eqref{eq:dominant_error_with_hCB}.

\begin{theorem}
\label{thm_high_snr_unoin_bound}
At high SNR values, i.e., $N_{o} << 1$, the term $\mbox{Pr}(error|h_{CB})$ in \eqref{eq:dominant_error_with_hCB} is upper bounded as
\begin{eqnarray}
\label{high_snr_union_bound}
\mbox{Pr}(error|h_{CB}) & \leq & \mbox{Pr}\left((0, 1) \rightarrow (1, 1) \right) + \nonumber \\
& & \mbox{Pr}\left((1, 1) \rightarrow (1, 2) \right) + \nonumber \\
& & \mbox{Pr}\left((1, 1) \rightarrow (0, 1) \right).
\end{eqnarray}
\end{theorem}
\begin{proof}
At high SNR values, for a given $j \in \{0, 1, \ldots, M-1\}$, the term $\mbox{Pr}((\hat{i}, \hat{j}) \neq (0, j) |(0, j))$ is upper bounded as
\begin{equation}
\mbox{Pr}((\hat{i}, \hat{j}) \neq (0, j) |(0, j)) \leq 2 \mbox{Pr}\left((0, j) \rightarrow (1, j) \right),
\end{equation}
wherein the rest of the terms in the union bound are neglected since their contributions are not dominant. Furthermore, due to the symmetry in the constellation, we have
\begin{equation}
\label{eq:intermediate_1}
\sum_{j = 0}^{M-1} \mbox{Pr}((\hat{i}, \hat{j}) \neq (0, j) |(0, j)) \leq 2M \left(\mbox{Pr}\left((0, 1) \rightarrow (1, 1) \right)\right).
\end{equation}
Similarly, for a given $j \in \{0, 1, \ldots, M-1\}$, the term $\mbox{Pr}((\hat{i}, \hat{j}) \neq (1, j) |(1, j))$ is upper bounded 
\begin{eqnarray}
\mbox{Pr}((\hat{i}, \hat{j}) \neq (1, j) |(1, j)) \leq 2 \mbox{Pr}\left((1, j) \rightarrow (0, j) \right) + \nonumber\\
 2 \mbox{Pr}\left((1, j) \rightarrow (1, (j+1) \mbox{ modulo } M) \right),
\end{eqnarray}
wherein the rest of the terms in the union bound are neglected since their contributions are not dominant. Furthermore, due to the symmetry in the constellation, we have
\begin{eqnarray*}
\label{eq:intermediate_2}
\sum_{j = 0}^{M-1} \mbox{Pr}((\hat{i}, \hat{j}) \neq (1, j) |(1, j)) \leq 2M \left(\mbox{Pr}\left((1, 1) \rightarrow (0, 1) \right)\right) + \\
2M \left(\mbox{Pr}\left((1, 1) \rightarrow (1, 2) \right)\right).
\end{eqnarray*}
Finally, by substituting \eqref{eq:intermediate_1} and \eqref{eq:intermediate_2} in \eqref{eq:dominant_error_with_hCB}, we get \eqref{high_snr_union_bound}. This completes the proof. 
\end{proof}

Based on Theorem \ref{thm_high_snr_unoin_bound}, three pair-wise error events dominate the error probability of the joint decoder at high SNR values. At lower values of $\alpha$, the error probability is dominated by $\mbox{Pr}\left((1, 1) \rightarrow (1, 2) \right)$, which is dictated by the intra-constellation symbols of the rotated $M$-PSK constellation; this is because the minimum distance of $\sqrt{\alpha}e^{\frac{\iota \pi}{M}}\mathcal{S}_{C}$ (denoted by $\ell$ in Fig. \ref{FFFD_CN_1}) is small and also the effective noise variance of the received symbol is very high. In contrast, as $\alpha$ starts to ascend, $\mbox{Pr}\left((0, 1) \rightarrow (1, 1) \right)$ and $\mbox{Pr}\left((1, 1) \rightarrow (0, 1) \right)$ dominate, which are dictated by the distance between symbols of $\sqrt{\alpha}e^{\frac{\iota \pi}{M}}\mathcal{S}_{C}$ and $\mathcal{S}_{C}$ (denoted by $d$ in Fig. \ref{FFFD_CN_1}). In the following proposition, we evaluate $\mbox{Pr}\left((0, 1) \rightarrow (1, 1) \right)$, $\mbox{Pr}\left((1, 1) \rightarrow (1, 2) \right)$ and $\mbox{Pr}\left((1, 1) \rightarrow (0, 1) \right)$ by using their definition in \eqref{eq:pair_wise_1} and \eqref{eq:pair_wise_2}.

\begin{proposition}
$\mbox{Pr}\left((0, 1) \rightarrow (1, 1) \right) = P_{00}P_{1} + P_{01}P^{c}_{1}$,
where 
\begin{eqnarray}
P_{1} & = & Q_1\left(\frac{|A|}{\sigma_{B0}},\frac{\sqrt{\xi}}{\sigma_{B0}}\right)\label{FFFD_P1},\\
P_1^c & = & Q_1\left(\frac{|B|}{\sigma_{B0}},\frac{\sqrt{\xi}}{\sigma_{B0}}\right),
\end{eqnarray}
where $Q_1(\cdot, \cdot)$ is the Marcum-Q function such that $A = \frac{\gamma d N_{B0}}{N_{B0}-N_{B1}}$, $d=\sqrt{(1 + \alpha - 2\sqrt{\alpha}\mbox{cos}\frac{\pi}{M})}$, $\gamma = |h_{CB}|$, $\xi = \frac{N_{B0}N_{B,1}}{N_{B0}-N_{B1}}\left[ln\left(\frac{N_{B0}P_{11}}{N_{B1}P_{00}}\right)+\frac{\gamma^2d^2}{N_{B0}-N_{B1}}\right]$, $\sigma_{B0} =\sqrt{\frac{N_{B0}}{2}}$, and $B = \frac{\gamma d N_{B1}}{N_{B0}-N_{B1}}$.
\end{proposition}

\begin{proposition}
$\mbox{Pr}\left((1, 1) \rightarrow (0, 1) \right) = P_{11}P_{2} + P_{10}P^{c}_{2}$
where
\begin{eqnarray}
P_2 & = & 1- Q_{1}\left(\frac{|B|}{\sigma_{B1}},\frac{\sqrt{\xi}}{\sigma_{B1}}\right),\\
P_2^{c} & = & 1-Q_1\left(\frac{|A|}{\sigma_{B1}},\frac{\sqrt{\xi}}{\sigma_{B1}}\right),
\end{eqnarray}
where $\sigma_{B1}=\sqrt{\frac{N_{B1}}{2}}$, in addition to the notations defined in the previous proposition. 
\end{proposition}

\begin{proposition}
$\mbox{Pr}\left((1, 1) \rightarrow (1, 2) \right) = P_{11}P_{3} + P_{10}P^{c}_{3}$
where
\begin{eqnarray}
P_3 & = & Q\left(\frac{\gamma \ell}{\sqrt{2N_{B1}}}\right),\\
P_3^{c} & = & 0.5,
\end{eqnarray}
where $\ell = 2\sqrt{\alpha}\sin{\frac{\pi}{M}}$, in addition to the notations defined in the previous propositions.
\end{proposition}

Using the expressions from the above propositions in \eqref{high_snr_union_bound}, $\mbox{Pr}(error|h_{CB})$ is bounded by 
\begin{equation}
\label{SER_J}
P_{00}P_1 +\!P_{01}P_1^c +P_{11}P_2 + P_{10}P_2^c + P_{11}P_3 + P_{10}0.5.
\end{equation}
Since $P_{01}$ is negligible for all values of $\alpha$, we replace the term $P_{00}P_1 +\!P_{01}P_1^c$ by $P_1$ in the above expression. Furthermore, note that $P_{11}, P_{10}, P_{00}$ are independent of $h_{CB}$, whereas $P_{1}, P_{2}, P^{c}_{2}, P_{3}$ are functions of $h_{CB}$. In the following proposition, we present the average probability of error of the joint dominant decoder, henceforth denoted as $\overline{P}_{e,Joint} = \mathbb{E}_{|h_{CB}|^2}[\mbox{Pr}(error|h_{CB})]$. 

\begin{proposition}
The average probability of error of the joint dominant decoder is upper bounded as
\begin{eqnarray}
\label{SER_Javg}
\overline{P}_{e,Joint} \leq P_{1, avg} + P_{11}\mathbb{E}_{|h_{CB}|^2}[P_{2}] + P_{10}\mathbb{E}_{|h_{CB}|^2}[P_{2}^{c}] \nonumber \\
+ ~P_{11}P_{3, avg} + ~P_{10}0.5,
\end{eqnarray}
where 
\begin{eqnarray*}
P_{1,avg} & = & \left(\frac{N_{B0}P_{11}}{N_{B1}P_{00}}\right)^\frac{N_{B1}}{N_{B1}-N_{B0}} \frac{(N_{B0}-N_{B1})^2}{(N_{B0}-N_{B1})^2 + d^2N_{B1}},\label{p1avg}\\
P_{3,avg} & = & \frac{2N_{B1}}{4N_{B1}+\ell^2}.\label{p3avg}
\end{eqnarray*}
\end{proposition}
\begin{proof}
Towards computing $\mathbb{E}_{|h_{CB}|^2}[P_{1}]$, we use the upper bound on the Marcum-Q function given by $P_{1} \leq e^{-\frac{1}{2}\left(\frac{\sqrt{\xi}}{\sigma_{B0}}-\frac{|A|}{\sigma_{B0}}\right)^2}$. Furthermore, we observe that $|A|<<\sqrt{\xi}$, and thus, simplifying the bound as $P_{1} \leq e^{-\frac{1}{2}\left(\frac{\sqrt{\xi}}{\sigma_{B0}}\right)^2}$. Finally, averaging it over $|h_{CB}|^{2}$, we get $\mathbb{E}_{|h_{CB}|^2}[P_{1}] < P_{1, avg}.$ Similarly, towards computing $\mathbb{E}_{|h_{CB}|^2}[P_{3}]$, we apply Chernoff-bound on the Q Function, and subsequently average it over $|h_{CB}|^2$ to get $\mathbb{E}_{|h_{CB}|^2}[P_{3}] < P_{3, avg}.$
\end{proof}

In the following result, we prove that $P_{1, avg}$ is not a dominant term of \eqref{SER_Javg} when the SNR is sufficiently large. 

\begin{theorem}\label{P1_P3}
Let $\varsigma = \mbox{cos}\frac{\pi}{M}$, where $M$ is the size of PSK constellation such that $(1-\varsigma^2)=\mu N_o$, where $\mu>>1$, then we have the inequality $P_{1,avg} < P_{11}P_{3,avg} + P_{10} 0.5$. 
\end{theorem}
\begin{proof}
We first prove $P_{1,avg} < P_{3,avg}$. In other words, we need to prove
\begin{equation*}
\frac{2N_{B1}}{4N_{B1}+\ell^2} \geq \left(\frac{N_{B0}P_{11}}{N_{B1}P_{00}}\right)^\frac{N_{B1}}{N_{B1}-N_{B0}} \frac{(N_{B0}-N_{B1})^2}{(N_{B0}-N_{B1})^2 + d^2N_{B1}}.
\end{equation*}
The expression for $P_{1,avg}$ can be further upper bounded as
\begin{equation}
P_{1,avg} \leq \left(\frac{N_{B0}P_{11}}{N_{B1}P_{00}}\right)^\frac{N_{B1}}{N_{B1}-N_{B0}} \frac{1-\alpha}{2-2\sqrt{\alpha}\mbox{cos}\frac{\pi}{M}}\label{p1_s1}.
\end{equation}
Let $1-\alpha = \rho N_o$, then (\ref{p1_s1}) becomes
\begin{small}
\begin{eqnarray}
P_{1,avg} & < & \left(\frac{P_{11}}{(\rho+1)P_{00}}\right)^{\frac{1+\rho}{\rho}} \frac{\rho N_o}{2-2\sqrt{1-\rho N_o}\mbox{cos}\frac{\pi}{M}}\label{p1_s3},\\
 & < & \left(\frac{1}{\rho+1}\right)^{\frac{1+\rho}{\rho}} \frac{\rho N_o}{2-2\sqrt{1-\rho N_o}\mbox{cos}\frac{\pi,}{M}}\label{p1_s2},
\end{eqnarray}
\end{small}
\noindent where the second inequality is applicable because $P_{11}<P_{00}, \forall \alpha$ from Lemma \ref{P01_P10}. Also, the expression for $P_{3,avg}$ is rewritten as
\begin{equation}
\label{p3_s1}
P_{3,avg} = \frac{(\rho+1)N_o}{N_o(2(1+\rho)-2\rho(1-\varsigma^2))+2(1-\varsigma^2)},
\end{equation}
where $\varsigma = \mbox{cos}\frac{\pi}{M}$. In the rest of the proof, we prove that either \eqref{p1_s3} or \eqref{p1_s2} is less than \eqref{p3_s1} considering three cases: $\rho = 1, \rho < 1, \mbox{ and } \rho > 1$.

For $\rho=1$, (\ref{p1_s2}) reduces to $\frac{1}{8}\frac{N_o}{1-\varsigma}$, by substituting $\rho=1$ and also upper bounding $(1-\rho N_o)$ by $1$. Also, (\ref{p3_s1}) can be simplified as $\frac{N_o}{N_o(1 + \varsigma^2)+(1-\varsigma^2)}$. Further, since $(1-\varsigma^2)=\mu N_o$, where $\mu>>1$, we have $N_o(1 + \varsigma^2) < 1-\varsigma^2$, and therefore, (\ref{p3_s1}) can be lower bounded by $\frac{N_o}{2(1-\varsigma^2)}$. Thus, we can easily conclude that $\frac{1}{8}\frac{N_o}{1-\varsigma}<\frac{N_o}{2(1-\varsigma^2)}$, $\forall \ 0 < \varsigma < 1$. 

For, $\rho<1$, we have the bound $\left(\frac{1}{1+\rho}\right)^{\frac{1+\rho}{\rho}}<\frac{1}{e}$, and therefore, (\ref{p1_s2}) is upper bounded as $\frac{1}{5.42}\frac{\rho N_o}{1-\varsigma}$. Further, we have $\frac{1}{5.42}\frac{\rho N_o}{1-\varsigma} = \frac{\rho(1+\varsigma)}{5.42\mu}<\frac{2\rho}{5.42\mu}$. Also, in the first term in the denominator of (\ref{p3_s1}), we have the bound $N_o(2(1+\rho)-2\rho(1-\varsigma^2)) < N_o(2(1+\rho))< 4N_o$. Thus, (\ref{p3_s1}) is lower bounded by $\frac{(\rho+1)}{2(2+\mu)}$. We can immediately infer that $\frac{(\rho+1)}{2(2+\mu)}>\frac{2\rho}{5.42\mu}$ when $\mu >> 1$ and for $\rho<1$. This completes the case of $\rho < 1$.

For, $\rho>1$, we split the case into two parts, namely, (i) $1 < \rho \leq 3$, and (ii) $\rho > 3$. For $1 < \rho\leq3$, the first term of denominator of (\ref{p3_s1}) is upper bounded as $8N_o$, and therefore (\ref{p3_s1}) is lower bounded as $\frac{\rho +1}{2(4+\mu)}$. Furthermore, upper bounding  (\ref{p1_s2}) on similar lines gives us $\left(\frac{1}{1+\rho}\right)^{\frac{1+\rho}{\rho}}\frac{\rho(1+\varsigma)}{2\mu}<\left(\frac{1}{1+\rho}\right)^{\frac{1+\rho}{\rho}}\frac{\rho}{\mu}$. Form this discussion, we can conclude that
\begin{IEEEeqnarray}{rCl}
\left(\frac{1}{1+\rho}\right)^{\frac{1+\rho}{\rho}}\frac{\rho}{\mu}&<&\frac{\rho +1}{2(4+\mu)},
\end{IEEEeqnarray}
wherein the above inequality holds good because $\left(\frac{1}{1+\rho}\right)^{\frac{1+\rho}{\rho}}$ is a decreasing function of $\rho$ and is upper bounded by $\frac{1}{4}$.\\

For $\rho>3$, if the condition $N_o(2(1+\rho)-2\rho(1-\varsigma^2))<2(1-\varsigma^2)$ holds, then (\ref{p3_s1}) is lower bounded as $\frac{(\rho+1)N_o}{4(1-\varsigma^2)} = \frac{1+\rho}{4\mu}$. Furthermore, (\ref{p1_s2}) is upper bounded as $\left(\frac{1}{1+\rho}\right)^{\frac{1+\rho}{\rho}}\frac{\rho N_o}{2(1-\varsigma)}$, in which the following inequalities $\left(\frac{1}{1+\rho}\right)^{\frac{1+\rho}{\rho}}<\left(\frac{1}{4}\right)^{\frac{4}{3}}$ and $\frac{\rho N_o}{2(1-\varsigma)}=\frac{\rho N_o(1+\varsigma)}{2(1-\varsigma^2)}<\frac{\rho}{\mu}$ hold. It is now straightforward to prove that $\frac{0.157\rho}{\mu}<\frac{0.25(\rho+1)}{\mu}$, for $0<\varsigma<1$ and $\rho>3$. Additionally, for the case $\rho>3$, and $\rho\leq\mu$, if we have $N_o(2(1+\rho)-2\rho(1-\varsigma^2))>2(1-\varsigma^2)$, then (\ref{p3_s1}) is lower bounded as $\frac{(1+\rho)N_o}{2N_o[2(1+\rho)-2\rho(1-\varsigma^2)]}>\frac{(1+\rho)N_o}{4N_o(1+\rho)}=\frac{1}{4}$. Therefore, $\frac{0.157\rho}{\mu}<\frac{1}{4}$. Now for $\rho>3$ and also, $\rho>\mu$, i.e., for larger values of $\rho$, we have $\left(\frac{1}{1+\rho}\right)^{\frac{1+\rho}{\rho}}<\frac{1}{1+\rho}$, and therefore, (\ref{p1_s2}) is upper bounded as $\frac{1}{1+\rho}\frac{\rho}{\mu}<\frac{1}{\mu}$. We can clearly see that, $\frac{1}{\mu}<\frac{1}{4}$, because $\mu>>1$. This completes the case of $\rho > 3$.

Finally, since $P_{1, avg}$ is also less than $0.5$, the statement of the theorem is proved because $P_{11}P_{3,avg} + P_{10} 0.5$ is a convex combination. This completes the proof. 
\end{proof}

Using Theorem \ref{P1_P3}, we can further upper bound $\overline{P}_{e,Joint}$ as $\overline{P}_{e,Joint}  \leq 2 P_{dom},$ where 
\begin{equation}
\label{SER_J1}
P_{dom} = P_{11}\left(P_{3,avg} + \mathbb{E}_{|h_{CB}|^2}[P_{2}]\right) + P_{10}\left(\mathbb{E}_{|h_{CB}|^2}[P_{2}^{c}] + \frac{1}{2}\right),
\end{equation}
such that $P_{dom}$ represents the dominant error events of the overall probability of error. Note that each term of \eqref{SER_J1} is a function of $\alpha$. Therefore, we are interested in characterizing the range of values of $\alpha$ in which one of the terms in $P_{dom}$ is significant than the others. This way, we can arrive at an appropriate value of $\alpha$ that minimizes the dominant error component of the joint probability of error of the FFFD JD decoder. 

\subsection{Domination of Error Events as a Function of $\alpha$}

To characterize the behavior of $P_{dom}$ as a function of $\alpha$, it is important to evaluate $\mathbb{E}_{|h_{CB}|^2}[P_{2}]$ and $\mathbb{E}_{|h_{CB}|^2}[P_{2}^{c}]$ in closed form. However, since both $P_{2}$ and $P_{2}^{c}$ are related of Marcum-Q functions, it is well known that exact expressions of $\mathbb{E}_{|h_{CB}|^2}[P_{2}]$ and $\mathbb{E}_{|h_{CB}|^2}[P_{2}^{c}]$ cannot be derived. On the other hand, while tight lower and upper bounds are available for $P_{2}^c$ \cite{MQF}, we notice that bounds are loose for $P_{2}$ \cite{MQF}, and as a result, the subsequent upper bounds on $P_{dom}$ will also be loose. Therefore, in this work, we do not take the conventional approach of minimizing $P_{dom}$ (or its upper bound) over $\alpha \in (0, 1)$. 

Applying numerical integration techniques to compute $\mathbb{E}_{|h_{CB}|^2}[P_{2}]$ and $\mathbb{E}_{|h_{CB}|^2}[P_{2}^{c}]$, we observe that $P_{11}\left(P_{3,avg} + \mathbb{E}_{|h_{CB}|^2}[P_{2}]\right)$ is a decreasing function of $\alpha$, whereas $P_{10}\left(\mathbb{E}_{|h_{CB}|^2}[P_{2}^{c}] + \frac{1}{2}\right)$ is an increasing function of $\alpha$. With that insight, computing the point of intersection between $P_{11}\left(P_{3,avg} + \mathbb{E}_{|h_{CB}|^2}[P_{2}]\right)$ and $P_{10}\left(\mathbb{E}_{|h_{CB}|^2}[P_{2}^{c}] + \frac{1}{2}\right)$ would give us a value of $\alpha$, say $\alpha = \alpha^{\dagger}$ below which the term $P_{11}\left(P_{3,avg} + \mathbb{E}_{|h_{CB}|^2}[P_{2}]\right)$ dominates the term $P_{10}\left(\mathbb{E}_{|h_{CB}|^2}[P_{2}^{c}] + \frac{1}{2}\right)$. Since the dominant term experiences a dip at $\alpha = \alpha^{\dagger}$, we can use $\alpha^{\dagger}$ as the power-splitting factor of the SC-FFFD technique.

Since $\mathbb{E}_{|h_{CB}|^2}[P_{2}]$ and $\mathbb{E}_{|h_{CB}|^2}[P_{2}^{c}]$ cannot be derived in closed-form, we cannot analytically evaluate the point of intersection $\alpha^{\dagger}$ in closed form. To circumvent this problem, we compute an approximation on $\alpha^{\dagger}$ by computing the point of intersection between a lower bound on $P_{11}\left(P_{3,avg} + \mathbb{E}_{|h_{CB}|^2}[P_{2}]\right)$ and a lower bound $P_{10}\left(\mathbb{E}_{|h_{CB}|^2}[P_{2}^{c}] + \frac{1}{2}\right)$. Towards that direction, the following proposition provides a tight lower bound on $\mathbb{E}_{|h_{CB}|^2}[P_{2}^{c}]$.

\begin{proposition} 
The term $\mathbb{E}_{|h_{CB}|^2}[P_{2}^{c}]$ satisfies the inequality 
\begin{equation}
\label{eqp2c_lb}
\mathbb{E}_{|h_{CB}|^2}[P_{2}^{c}] > P_{2, avg}^{c} = \frac{4d^{2}N^{2}_{o}}{4d^{2}N^{2}_{o} + (N_{o} + (1-\alpha))(1 - \alpha)^{2}}.
\end{equation}
\end{proposition}
\begin{proof}
We apply the lower bound on $P_{2, avg}^{c}$ by using an upper bound on the Marcum-Q function which results in $P_{2}^c \geq 1-e^{-\frac{1}{2}\left(\frac{\sqrt{\xi}}{\sigma_{B1}}-\frac{|A|}{\sigma_{B1}}\right)^2}$. Subsequently, we notice that $\sqrt{\xi}>3|A|$, and therefore simplify the bound to $P_{2}^c\geq 1-e^{-2\left(\frac{|A|}{\sigma_{B1}}\right)^2}$ Finally, we average this bound over $|h_{CB}|^2$ to obtain \eqref{eqp2c_lb}.
\end{proof}

In addition to the bound in \eqref{eqp2c_lb}, we also observe that $\mathbb{E}_{|h_{CB}|^2}[P_{2}] > 0$ trivially. Using these two lower bounds, we are interested in computing the range of values of $\alpha$ in which $P_{11}P_{3,avg}$ dominates the term $P_{10}\left(P_{2, avg}^{c} + \frac{1}{2}\right)$. Note that both these terms are in closed form, and as a result, the point of intersection between the two can be computed analytically. To assist computing the dominant term between the two, we show in Lemma \ref{L_p3_dec} that $P_{3,avg}$ is a decreasing function of $\alpha$, and also show in Lemma \ref{P2c_inc} that $\left(P_{2, avg}^c + \frac{1}{2}\right)$ is an increasing function of $\alpha$. Furthermore, given that $P_{11}$ and $P_{10}$ are decreasing and increasing functions of $\alpha$, respectively, we show that computing the value of $\alpha$ at which $P_{11}P_{3,avg}$ intersects with $P_{10}\left(P_{2, avg}^c + \frac{1}{2}\right)$ gives the range of values of $\alpha$ for which $P_{11}P_{3,avg}$ dominates the term $P_{10}\left(P_{2, avg}^{c} + \frac{1}{2}\right)$. 

\begin{lemma}\label{L_p3_dec}
The term $P_{3,avg}$ decreases as $\alpha$ increases in the interval $(0, 1)$.
\end{lemma} 
\begin{proof}
The expression for $P_{3,avg}$ is given by
\begin{equation*}
P_{3,avg} = \frac{2N_{B1}}{4N_{B1}+\ell^2} = \frac{2(N_o + 1 -\alpha)}{4(N_o+1-\alpha) + 4\alpha sin^{2}\left(\frac{\pi}{M}\right)}.
\end{equation*}
Differentiating the above w.r.t. $\alpha$, we get
$\frac{dP_{3,avg}}{d\alpha} =\frac{1}{2}\frac{\left[-N_o - 1 + \alpha\left(1 - sin^2\left(\frac{\pi}{M}\right)\right)\right]+ \left[(N_o + 1 -\alpha)\left(1-sin^2\left(\frac{\pi}{M}\right)\right)\right]}{\left[N_o + 1-\alpha \left(1-sin^2\left(\frac{\pi}{M}\right)\right)\right]^2}$.
Closely observing the above equation revels that the denominator is always a positive quantity, whereas the numerator can be simplified to obtain $-(N_o + 1)sin^2\left(\frac{\pi}{M}\right)$. Thus, $\frac{dP_{3,avg}}{d\alpha}$ is always negative. Therefore, $P_{3,avg}$ is a decreasing function w.r.t. $\alpha$.
\end{proof}

\begin{lemma}\label{P2c_inc}
Let $\varsigma = \mbox{cos}(\frac{\pi}{M})$, where $M$ is the size of the PSK constellation such that $(1 - \varsigma^{2}) = \mu  N_{o}$, where $\mu >> 1$. When $\alpha \in (\varsigma^{2}, 1)$, $P_{2,avg}^c$ is an increasing function of $\alpha$, and when $\alpha \in (0, \varsigma^{2})$, we have $P_{2,avg}^c + 0.5 \approx 0.5$.
\end{lemma}
\begin{proof}
When $\alpha \in (\varsigma^{2}, 1)$, it is straightforward to observe that $d^{2}$ is an increasing function of $\alpha$. Also, the term $P_{2,avg}^c$ can be rewritten as
\begin{equation*}
\frac{4N^{2}_{o}}{4N^{2}_{o} + \frac{(N_{o} + (1-\alpha))(1 - \alpha)^{2}}{d^{2}}}.
\end{equation*}
Since $d^{2}$ is an increasing function in $\alpha \in (\varsigma^{2}, 1)$, and the term $(N_{o} + (1-\alpha))(1 - \alpha)^{2}$ is a decreasing function of $\alpha \in (0, 1)$, the fraction $\frac{(N_{o} + (1-\alpha))(1 - \alpha)^{2}}{d^{2}}$ is a decreasing function of $\alpha$ in the range $\alpha \in (\varsigma^{2}, 1)$. This completes the proof that $P_{2,avg}^c$ is an increasing function of $\alpha$ when $\alpha \in (\varsigma^{2}, 1)$. For the second part, the term $P_{2,avg}^c$ is upper bounded as
\begin{equation*}
\frac{4d^{2}N^{2}_{o}}{(N_{o} + (1-\alpha))(1 - \alpha)^{2} + 4\mu N^{3}_{o}},
\end{equation*}
by using the lower bound $d^{2} \geq 1 - \varsigma^{2} = \mu N_{o}$ in the second term of the denominator. Furthermore, we can upper bound $d^{2}$ in the numerator by $1-\alpha$ by using the lower bound $\varsigma > \sqrt{\alpha}$ in the range $\alpha \in (0, \varsigma^{2})$. Let us also denote $(1 - \alpha) = \rho N_{o}$, where $\rho > 0$. With that the upper bound can now be written as 
\begin{equation*}
\frac{4 \rho N^{3}_{o}}{N^{3}_{o}(\rho^{3} + \rho^{2}) + 4\mu N^{3}_{o}} = \frac{4\rho}{(\rho^{3} + \rho^2 + 4\mu)}, 
\end{equation*}
where $\rho \geq \mu$ since $\alpha \leq \varsigma^{2}$ and $1 - \varsigma^{2} = \mu N_{o}$. Finally, since $\rho \geq \mu >> 1$, the above term is a negligible number, and therefore, $P_{2,avg}^c + 0.5 \approx 0.5$. This completes the proof for the second part. 
\end{proof}

From Lemma \ref{L_p11}, Lemma \ref{L_p3_dec}, and Lemma \ref{P2c_inc}, we deduce that $P_{11}P_{3,avg}$ decreases with $\alpha$, whereas $P_{10}(P_{2,avg}^c+0.5)$ increases with $\alpha$. With that the following theorem shows that $P_{11}P_{3,avg}$ and $P_{10}(P_{2,avg}^c+0.5)$ intersect only once in the interval $(0, 1)$

\begin{theorem}\label{P2P3}
When SNR is sufficiently large, the terms $P_{11}P_{3,avg}$ and $P_{10}(P_{2,avg}^c+0.5)$ intersect only at one value of $\alpha$, say $\alpha^{*}$, in the interval $(0, 1)$.
\end{theorem}
\begin{proof}
Let $f_1(\alpha)=P_{11}P_{3,avg}$ and $P_{10}(P_{2,avg}^c+0.5)$. Evaluating the extreme values of $f_1(\alpha)$ and $f_2(\alpha)$, we get 
\begin{eqnarray*}
f_1(0) & = & \left.P_{11}P_{3,avg}\right|_{\alpha=0}=\frac{1}{2},\\
f_2(0) & = & \left.P_{10}(P_{2,avg}^c+0.5)\right|_{\alpha=0} < \frac{1}{2},
\end{eqnarray*}
where the second inequality applies since $P_{10} |_{\alpha = 0} << 1$ and $P_{2,avg}^c+0.5 \approx 0.5$ for $\alpha = 0$ (from Lemma \ref{P2c_inc}). Similarly,
\begin{eqnarray*}
f_1(1) & = & \left.P_{11}P_{3,avg}\right|_{\alpha=1}=N_oe^{-1},\\
f_2(1) & = & \left.P_{10}(P_{2,avg}^c+0.5)\right|_{\alpha=1}\approx \frac{3}{2}(1-e^{-1}).
\end{eqnarray*}
Finally, we define $f(\alpha) \triangleq f_{1}(\alpha) - f_{2}(\alpha)$. Since $f_1(0) > f_2(0)$ and $f_1(1) < f_2(1)$, we have $f(0) > 0$ and $f(1) < 0$. In addition, since $f(\alpha)$ is a decreasing function of $\alpha$, it implies that $f(\alpha)$ has a unique root. Therefore, $f(\alpha^{*}) = 0$ for some $\alpha^{*} \in (0, 1)$. This completes the proof. 
\end{proof}

With $\alpha^{*}$ being the point of intersection between $P_{11}P_{3,avg}$ and $P_{10}(P_{2,avg}^c+0.5)$, we propose to use the value of $\alpha^{*}$ as the power-splitting factor between Alice and Charlie to implement the SC-FFFD relaying scheme. In practice, we can use the well known Newton-Raphson algorithm (NR) \cite{NR} to compute the root of $f(\alpha) = P_{11}P_{3,avg} - P_{10}(P_{2,avg}^c+0.5)$ as a function of the $M$-PSK constellation, noise variance $N_{o}$, and $\sigma^{2}_{AC}$. 


\section{Simulation Results}
\label{sec:sims}

In this section, we present simulation results to showcase the effectiveness of the proposed SC-FFFD technique to mitigate the jamming attack by a FD adversary. Throughout this section, we use the system model in Section \ref{sec:system_model} wherein the channels are distributed as $h_{AB} \sim \mathcal{CN}(0, 1)$, $h_{CB} \sim \mathcal{CN}(0, 1)$, and $h_{AC} \sim \mathcal{CN}(0, 4)$. We specifically choose $\sigma^{2}_{AC} = 4$ to showcase the benefits of the SC-FFFD technique when the channel between Alice and Charlie is more reliable than that between Alice (or Charlie) and Bob. We also use $\mbox{SNR } = \frac{1}{N_{o}}$ throughout this section. First, to present the variation of the error performance of the SC-FFFD technique with the power-splitting factor $\alpha \in (0, 1)$, we plot the average probability of error of various joint decoders in Fig. \ref{cumla} as a function $\alpha$ at SNR = 35 dB, and with $4$-, and $8$-PSK at Charlie. We use Monte-Carlo simulations to plot the average probability of error of the decoding metrics in \eqref{eq:joint_MAP_Bob} (the joint MAP decoder), \eqref{eq:joint_MAP_Bob_max} (the joint MAX decoder), and \eqref{eq:joint_MAP_Bob_dominant} (the joint dominant decoder). However, to plot the union bound in \eqref{SER_Javg}, we have used a combination of analytical expressions and numerical integration techniques. The plots in Fig. \ref{cumla} show that the curves decrease as a function of $\alpha$ upto a certain point, and then shoots up as $\alpha$ approaches $1$. This behavior of the curves is very intuitive as $\alpha = 0$ signifies that Charlie is not sending any message. Therefore, since Bob jointly decodes the information symbols of Alice and Charlie, he would have to guess Charlie's symbol, thereby resulting in degraded error performance. Similarly, when $\alpha=1$, Bob would have to guess Alice's symbol, and that explains the steep rise in the error. Between these extreme values, as $\alpha$ increases, the performance gradually improves because Charlie injects more power for its symbols. We also plot the union bound in \eqref{SER_Javg} to compare it with the average probability of error of the joint decoders. It can be observed that the dip in $\alpha$ for all the three decoders are very close to each other.

\begin{figure}[t]
\begin{center}
\includegraphics[scale=0.48]{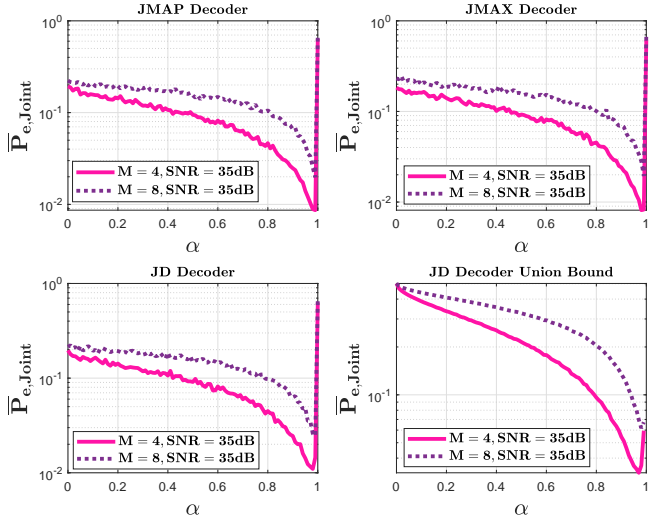}
\vspace{-0.8cm}
\caption{\label{cumla}Average probability of error of various joint decoders of the SC-FFFD technique as a function of $\alpha \in (0, 1)$ at SNR = 35 dB. The plots show that the minimum probability of error is achieved at a value of $\alpha$ close to $\alpha = 1$. Similar behavior is also observed with the bound in \eqref{SER_Javg}.}
\end{center}
\end{figure}

\begin{figure}[t]
\centerline{\includegraphics[scale=0.4]{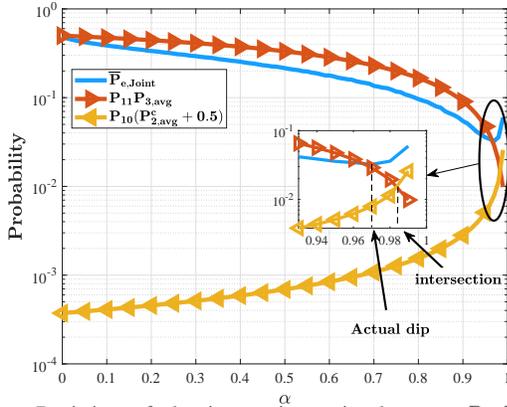}}
\vspace{-0.5cm}
\caption{\label{ub}Depiction of the intersection point between $P_{11}P_{3, avg}$ and $P_{10}(P_{2, avg}^{c} + 0.5)$ as a function of $\alpha \in (0, 1)$ at SNR = 35 dB and $4$-PSK at Charlie. The point of intersection is approximately close to the minima of the bound in \eqref{SER_Javg}.}
\end{figure}
\begin{figure}[t]
\centerline{\includegraphics[scale=0.43]{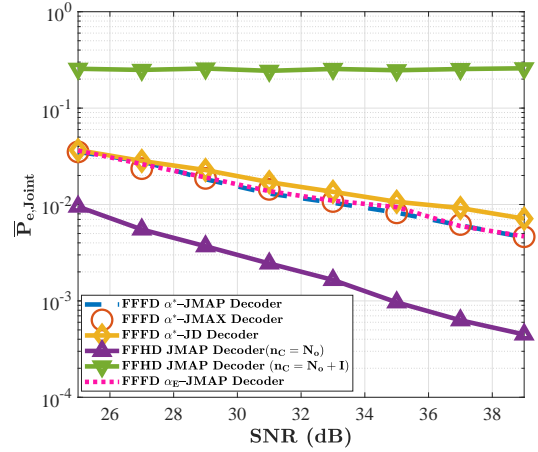}}
\vspace{-0.5cm}
\caption{\label{alp_cmp}Average probability of error of several joint decoders against various SNR values. The proposed FFFD joint dominant decoder is able to drive down the probability of error with increasing SNR while providing low-complexity analytical solutions to derive the power-splitting ratio.}
\end{figure}
\begin{figure}[t]
\centerline{\includegraphics[scale=0.43]{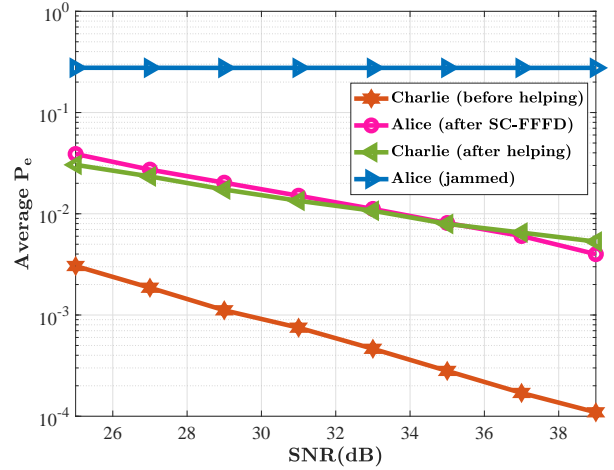}}
\vspace{-0.2cm}
\caption{\label{tf}Symbol-error-probability of information symbols of Alice and Charlie with and without using the SC-FFFD technique against different SNR values. It is clear from the plots that Charlie bails out Alice from the jamming attack at the cost of degradation in its error performance.}
\end{figure}

As one of the main contributions of this work, we provide an analytical approach to compute the value of $\alpha$ at which the average probability of error of the joint dominant decoder experiences a dip over $\alpha \in (0, 1)$. As explained in Section \ref{sec:subopt_decoder}, we propose to solve the intersection point between $P_{11}P_{3,avg}$ and $P_{10}(P_{2,avg}^c+0.5)$, denoted by $\alpha^*$, using the well-known Newton-Raphson (NR) algorithm. To depict the closeness between this intersection point and the minima of the union bound, we plot both of them in Fig. \ref{ub} as a function of $\alpha$ at SNR = 35 dB, and with $4$-PSK at Charlie.

When using $\alpha^*$ as the power-splitting factor between Alice and Charlie, in Fig. \ref{alp_cmp}, we plot the error performance of FFFD $\alpha^*$-JMAP decoder, FFFD $\alpha^*$-JMAX decoder, and FFFD $\alpha^*$-JD decoder as a function of SNR. In this context, we have prefixed $\alpha^{*}$ with JMAP and JMAX variants of the decoder to highlight that for each SNR value, the corresponding value of $\alpha^{*}$ is obtained from the NR algorithm. The plots show that using FFFD $\alpha^*$-JMAP decoder provides the best error performance among the three decoders. Furthermore, we also present the error performance of the FFFD $\alpha_E$-JMAP decoder, wherein the value of $\alpha_E \in (0, 1)$ minimizes the average probability of error of the JMAP decoder. We have computed $\alpha_{E}$ through exhaustive search over the interval $(0, 1)$ in steps of 0.001. The plots also show that the FFFD $\alpha_E$-JMAP decoder provides error performance very close to that of the $\alpha^*$-JMAP decoder. However, unlike the FFFD $\alpha^{*}$-JMAP decoder, the best value of $\alpha_{E}$ can only be computed using exhaustive search through simulations. As a result, applying the FFFD $\alpha_E$-JMAP decoder in practice is prohibitively complex. As a competitive baseline for the SC-FFFD technique, we have also considered an alternate cooperative relaying technique, wherein Alice continues to transmit her OOK symbols on the frequency band $f_{AB}$. Meanwhile, Charlie, which works in the half-duplex mode, listens to Alice's symbol by tuning to the frequency band $f_{AB}$, decodes it, and then instantaneously rotates its chosen PSK symbol by either $\frac{\pi}{N}$ or $0$ radians, depending on whether the decoded bit is $1$ or $0$, respectively. Finally, the modified PSK symbol is transmitted to Bob on the frequency band $f_{CB}$. Note that this scheme does not involve any power-splitting factor since Alice continues to communicate on the frequency band $f_{AB}$. Referring to this scheme as the Fast-Forward Half-Duplex (FFHD) technique, we also plot the average probability of error of the joint MAP decoder (denoted by FFHD JMAP decoder) in Fig. \ref{alp_cmp} under two scenarios: (i) when the location of Dave is such that the jamming energy on the frequency band $f_{AB}$ does not reach Charlie's receiver, and (ii) when the location of Dave is such that the jamming energy on the frequency band $f_{AB}$ reaches Charlie's receiver as significant interference (denoted by $\mathbf{I}$). The plots show that the FFHD scheme in the former scenario outperforms the proposed SC-FFFD relaying scheme, whereas the FFHD scheme in the latter scenario is not a favorable choice.  

Finally, we discuss the trade-off offered by the SC-FFFD scheme in improving the error performance of Alice's communication at the cost of degrading the error performance of Charlie's communication. In Fig. \ref{tf}, we plot the average symbol error probability in decoding Alice's and Charlie's information symbols before and after executing the SC-FFFD technique. To generate these plots, we use $4$-PSK at Charlie and OOK at Alice both before and after the SC-FFFD technique. Based on the plots in Fig. \ref{tf}, it is intuitive that after using $\alpha^*$ as the power-splitting factor of the SC-FFFD scheme, Alice's performance improves drastically as its information symbols are encoded in the form of rotation of PSK constellation, as well as the noise variance of the effective noise. However, it is observed that Charlie's performance deteriorates because Bob has to now make a decision between $8$ PSK symbols to decode Charlie's information symbols. Overall, we highlight that the proposed cooperative relaying strategy serves the purpose of forcing the attacker to continue executing the DOS attack on $f_{AB}$, while making sure that none of the other nodes in the network experience DOS attacks. 

\section{Conclusion}
\label{sec:conc}

In this paper, we have presented a novel cooperative mitigation strategy, referred to as the SC-FFFD scheme to facilitate communication of low-latency packets in the presence of a full-duplex adversary. We have observed that although the helper node takes a hit in its error performance, the victim node can reliably communicate its packets to the destination. Moreover, the two nodes cooperatively inject power on the jammed frequency so as to keep the adversary engaged on the jammed frequency band. As one of the main contributions of this work, we have analyzed the error performance of jointly decoding the information symbols of the victim and the helper node when they employ the OOK and PSK modulation, respectively. Our analysis has shown that an appropriate value of the power-splitting factor can be analytically computed by observing the dominant error events of the average probability of error of the joint decoder. We strongly believe that the proposed solutions are effective in scenarios wherein (i) the number of frequency bands to hop is limited compared to the number of users in the network, and (ii) there exists users equipped with full-duplex radios to assist the victim node.

\bibliographystyle{IEEEtran}
\bibliography{IEEEabrv, Reference}

\end{document}